\documentclass[11pt,a4paper]{article}
\usepackage{subfig}
\usepackage{graphicx}
\usepackage{amsthm,amsmath,amssymb,mathrsfs}
\usepackage{tikz}
\usetikzlibrary{decorations.pathreplacing}
\usetikzlibrary{calc}
\newcommand{\ep}{\epsilon}

\newtheorem{theorem}{Theorem}
\newtheorem{lemma}{Lemma}

\begin{document}
\title{A Framework for Vehicle Routing Approximation Schemes in Trees}


\author{Amariah Becker\thanks{Research funded by NSF grant CCF-14-09520}\\
Computer Science Department, Brown University\\  
\texttt{amariah\_becker@brown.edu}
\and Alice Paul\\
Data Science Initiative, Brown University\\
\texttt{alice\_paul@brown.edu}}


\maketitle
\begin{abstract}
We develop a general framework for designing polynomial-time approximation schemes (PTASs) for various vehicle routing problems in trees.  In these problems, the goal is to optimally route a fleet of vehicles, originating at a depot, to serve a set of clients, subject to various constraints.  For example, in {\sc Minimum Makespan Vehicle Routing}, the number of vehicles is fixed, and the objective is to minimize the longest distance traveled by a single vehicle.  Our main insight is that we can often greatly restrict the set of potential solutions without adding too much to the optimal solution cost. This simplification relies on partitioning the tree into clusters such that there exists a near-optimal solution in which every vehicle that visits a given cluster takes on one of a few forms. In particular, only a small number of vehicles serve clients in any given cluster.  By using these coarser building blocks, a dynamic programming algorithm can find a near-optimal solution in polynomial time. We show that the framework is flexible enough to give PTASs for many problems, including {\sc Minimum Makespan Vehicle Routing}, {\sc Distance-Constrained Vehicle Routing}, {\sc Capacitated Vehicle Routing}, and {\sc School Bus Routing}, and can be extended to the multiple depot setting.
\end{abstract}

\section{Introduction}\label{sec:intro}

Vehicle routing problems address the fundamental problem of routing a fleet of vehicles from a common depot to visit a set of clients. These problems arise naturally in many real world settings, and are well-studied across computer science and operations research.  We generalize a class of vehicle routing problems by introducing the notions of \emph{vehicle load}, the problem-specific vehicle constraint (e.g. number of clients visited, distance traveled by the vehicle, client regret, etc.), and \emph{fleet budget}, the problem-specific fleet constraint (e.g. number of vehicles, sum of distances traveled, etc.).  

Most vehicle routing problems can then be framed as either {\sc Min-Max Vehicle Load}: minimize the maximum vehicle load, given a bound $k$ on fleet budget (e.g. {\sc Minimum Makespan Vehicle Routing}) or {\sc Minimum Fleet Budget}: minimize the required fleet budget, given a bound $D$ on vehicle load (e.g. {\sc Distance-Constrained Vehicle Routing}).  In fact, these are two optimization perspectives of the same decision problem: does there exist a solution with maximum vehicle load $D$ and fleet budget $k$?

\subsection{Main Contributions}

We present a framework for designing polynomial time approximation schemes (PTASs) for {\sc Min-Max Vehicle Load} and {\sc Minimum Fleet Budget} in trees.  Tree (and treelike) transportation networks occur in  building and warehouse layouts, mining and logging industries, and along rivers and coastlines~\cite{labbe,karuno97}. Our framework applies directly to {\sc Min-Max Vehicle Load} problems and generates results of the following form.

\begin{theorem}\label{thm:main}
For every $\ep > 0$, there is a polynomial-time algorithm that, given an instance of {\sc Min-Max Vehicle Load} on a tree, finds a feasible solution whose maximum vehicle load is at most $1+\ep$ times optimum.
\end{theorem}

An immediate corollary of Theorem~\ref{thm:main} is the following \emph{bicriteria} result for the associated {\sc Minimum Fleet Budget} problem.

\begin{theorem}
Given an instance of {\sc Minimum Fleet Budget} on a tree, if there exists a solution with fleet budget $k$ and vehicle load $D$, then for any $\ep > 0$, there is a polynomial-time algorithm that finds a solution with fleet budget $k$ and vehicle load at most $(1+\ep)D$.
\end{theorem}

The input to the framework is a rooted tree $T = (V,E)$ with root $r \in V$ and edge lengths $\ell(u,v) \geq 0$ for all $(u,v) \in E$. The root $r$ represents the \emph{depot} at which all vehicles start. Without loss of generality, the set of clients corresponds to the set of leaves in the tree (any subtree without a client can be safely removed from the instance).  Since every edge must then be traversed by at least one vehicle, the problems are equivalent to corresponding tree-cover problems.

As stated, the framework can be customized to a wide range of problems.  In Section~\ref{sec:makespan}, we illustrate in detail how to customize the framework to give a PTAS for the {\sc Minimum Makespan Vehicle Routing} problem of finding $k$ \emph{tours} each starting and ending at a depot $r$ that serve all clients in $T$ such that the \emph{makespan}, the maximum length of any tour, is minimized.  Here, vehicle load is the tour length, and fleet budget is the number of vehicles.  A bicriteria PTAS for the associated {\sc Minimum Fleet Budget} problem, {\sc Distance-Constrained Vehicle Routing}, follows as a corollary.  

We will also show how the framework can be applied to give similar results for other vehicle-routing variants, including {\sc Capacitated Vehicle Routing} (see Section~\ref{sec:capacity}) and {\sc School Bus Routing} (see Section~\ref{sec:school_bus}).  Additionally, we show how to generalize to the multiple depot setting (see Section~\ref{sec:multi}). The breadth of the problems listed highlights the real flexibility and convenience of the presented framework.

At a high level, the framework partitions the tree into \emph{clusters} such that there exists a near-optimal solution that within each cluster has a very simple form, effectively coarsening the solution space.  Then, given this simplified structure, a dynamic program can be designed to find such a near-optimal solution.

The clusters are designed to be \emph{small} enough so that simplifying vehicle routes at the cluster level does not increase the optimal load by too much, but also \emph{large} enough that the (coarsened) solutions can be enumerated efficiently.  To bound the error introduced by this simplification we design a load-reassignment tool that makes cluster coverage adjustments \emph{globally} in the tree.

Finally, standard dynamic programming techniques can result in a large accumulation of rounding error.  To limit the number of times that the load of any single route is rounded, we introduce a \emph{route projection} technique that essentially pays in advance for load that the vehicle \emph{anticipates} accumulating, allowing the dynamic program to round only once instead of many times for this projected load.

\subsection{Related Work}

For trees, {\sc Minimum Makespan Vehicle Routing} is equivalent to {\sc Minimum Makespan Rooted Tree Cover}: the minimum makespan for rooted tree cover is exactly half the minimum makespan for vehicle routing, since tours traverse edges twice.  {\sc Minimum Makespan Rooted Tree Cover} is NP-hard even on star instances but admits an FPTAS if the number, $k$, of subtrees is constant~\cite{sahni} and a PTAS for general $k$~\cite{hochbaum}.  For covering a \emph{general graph} with rooted subtrees, \cite{even} provides a 4-approximation; this bound was later improved to a 3-approximation by~\cite{nagamochi}.  For tree metrics, an FPTAS is known for constant $k$~\cite{xu}, and a $(2+\ep)$-approximation is known for general $k$~\cite{nagamochi}.  In this paper, we improve this to a PTAS.  Although a recent paper of Chen and Marx~\cite{chen_marx} also claimed to present a PTAS, in Appendix~\ref{app:chen_marx} we show that their result is incorrect and cannot be salvaged using the authors' proposed techniques.  Additionally, we compare their approach to our own and describe how we successfully overcome the challenges where their approach fell short.

The associated {\sc Distance-Constrained Vehicle Routing} problem is to minimize the number of tours of length at most $D$ required to cover all client demand.  Even restricted to star instances, this problem is NP-hard, and for tree instances it is hard to approximate to better than a factor of $3/2$~\cite{nagarajan}.  A 2-approximation is known for tree instances, and $O(\log D)$ and $O(\log |S|)$-approximations are known for general metrics, where $S$ is the set of clients~\cite{nagarajan}.  Allowing a multiplicative stretch in the distance constraint, a $(O(\log1/\ep), 1+\ep)$ bicriteria approximation is also known, which finds a solution of at most $O(\log1/\ep)OPT_D$ tours each of length at most $(1+\ep)D$~\cite{nagarajan}, where $OPT_D$ is the minimum number of tours of length at most $D$ required to cover all clients.  We give a $(1, 1+\ep)$ bicriteria PTAS for trees, noting that a true PTAS is unlikely to exist~\cite{nagarajan}.

In the classic {\sc Capacitated Vehicle Routing} each vehicle can cover at most $Q$ clients, and the objective is to minimize the \emph{sum} of tour lengths.  This problem is also NP-hard, even in star instances~\cite{labbe}.  For tree metrics, a 4/3-approximation is known~\cite{becker}, which improves upon the previous best-known approximation ratio of $(\sqrt{41}-1)/4$ by~\cite{asano} and is tight with respect to the best known lower bound.  In this paper, we give a $(1, 1+\ep)$ bicriteria PTAS for trees.  For general metrics, a $(2.5-\frac{1.5}{Q})$-approximation is known~\cite{haimovich}(using \cite{christofides}).

The \emph{regret} of a path is the difference between the path length and the distance between the path endpoints.  The {\sc Min-Max Regret Routing} problem is to cover all clients with $k$ \emph{paths} starting from the depot, such that the maximum regret is minimized. For trees, there is a known 13.5-approximation algorithm~\cite{bock}, which we improve to a PTAS in this paper.  For general graphs there is a $O(k^2)$-approximation algorithm~\cite{friggstad14}.  

In the related {\sc School Bus Routing} problem, there is a bound $R$ on the regret of each path and the goal is to find the minimum number of paths required to cover all client demand. For general graphs,~\cite{friggstad17} provide an LP-based 15-approximation algorithm, improving upon the authors' previous 28.86-approximation algorithm~\cite{friggstad14}.  In trees, there exists a 3-approximation algorithm for the uncapacitated version of this problem and a 4-approximation algorithm for the capacitated version~\cite{bock}. Additionally, there is a (3/2) inapproximability bound \cite{bock}. A true PTAS is therefore unlikely to exist for trees. Instead, we give a $(1, 1+\ep)$ bicriteria PTAS.

\section{Preliminaries}\label{sec:preliminaries}

Let $OPT$ denote the value of an optimum solution. For a minimization problem, a polynomial-time $\alpha$-approximation algorithm is an algorithm that finds a solution of value at most $\alpha\cdot OPT$ and runs in time that is polynomial in the size of the input.  A polynomial-time approximation scheme (PTAS) is a family of $(1+\ep)$-approximation algorithms indexed by $\ep> 0$ such that for each $\ep$, the algorithm runs in time polynomial in the input size, but may depend arbitrarily on $\ep$. 

In a rooted tree, the \emph{parent} of a vertex $v$, denoted $p(v)$, is the vertex adjacent to $v$ in the shortest path from $v$ to $r$ (the parent of $r$ is undefined).  If $u = p(v)$ then $v$ is a \emph{child} of $u$.  The parent edge of a vertex $v$ is the edge $(p(v),v)$ (undefined for $v=r$). The \emph{ancestors} of vertex $v$ are all vertices (including $v$ and $r$) in the shortest $v$-to-$r$ path and the \emph{descendants} of $v$ are all vertices $u$ such that $v$ is an ancestor of $u$.  We assume every vertex has at most two children. If vertex $v$ has $l>2$ children $v_1,...,v_l$, add vertex $v'$ and edge $(v,v')$ of length zero and replace edges $(v,v_1)$,$(v,v_2)$ with edges $(v',v_1)$,$(v',v_2)$ of the same lengths.

Further, the \emph{subtree rooted at $v$} is the subgraph induced by the descendants of $v$ and is denoted $T_v$.  If $u = p(v)$, the $v$-\emph{branch} at $u$ consists of the subtree rooted at $v$ together with the edge $(u,v)$. We define the \emph{length} of a subgraph $A\subseteq E$ to be $\ell(A) = \sum_{(u,v) \in A} \ell(u,v)$.  For vertices $u,v$, we use $d_T(u,v)$ to denote the shortest-path distance in $T$ between $u$ and $v$.


Our framework applies to vehicle routing problems that can be framed as a {\sc Min-Max Vehicle Load} problem, in which the objective is to minimize the maximum vehicle load, subject to a fleet budget.  Given a {\sc Min-Max Vehicle Load} problem, a trivial $n$-approximation can be used to obtain an upper bound $D_{high}$ for $OPT$.  An overarching algorithm  takes as input a load value $D\geq 0$ and provides the following guarantee: if there exists a solution with max load $D$, the algorithm will find a solution with max load at most $(1+\ep)D$.  A PTAS follows from using binary search between $\frac{D_{high}}{n}$ and $D_{high}$ for the smallest value $D_{low}$ such that the algorithm returns a solution of max load at most $(1+\ep)D_{low}$. This implies $D_{low} \leq OPT$.  For the rest of the paper, we assume $D$ is fixed.

\section{Framework Overview}\label{sec:framework}

Optimization problems on trees are often well suited for dynamic programming algorithms. In fact, the following dynamic programming strategy can solve {\sc Min-Max Vehicle Load} problems on trees exactly: at each vertex $v$, for each value $0\leq i \leq D$, \emph{guess} the number of solution route segments of load exactly $i$ in the subtree rooted at $v$.  Such an algorithm would be exponential in $D$.  Instead of considering every possible load value, route segment loads can be \emph{rounded} up to the nearest $\theta D$, for some value $\theta \in (0,1]$ that depends only on $\ep$, so that only $O(\theta^{-1})$ segment load values need to be considered.  In order to achieve a PTAS, we must show that this rounding does not incur too much error. Rounding the load of a route at \emph{every} vertex accumulates too much error, but if the number of times that any given route is rounded is at most $\ep/\theta$, then at most $\ep D$ error accumulates, as desired. 

One main insight underlying our algorithm is that a route only needs to incur rounding error when it branches.  The challenge in bounding the rounding error then becomes bounding the number of times a route branches. While a route in the optimal solution may have an arbitrary amount of branching, we show that we can greatly limit the scope of candidate solutions to those with a specific structure while only incurring an $\ep D$ error in the maximum load.  Rather than having to make decisions for covering every leaf in the tree (of which there may be arbitrarily many$-$each with arbitrarily small load), we partition the tree into \emph{clusters} and then address covering the clusters.  

By \emph{reassigning} small portions of routes within a cluster, we show that there exists a near-optimal solution in which all clients (leaves) within a given cluster are covered by only one or two vehicles.  These clusters are chosen to be small enough that the error incurred by the reassignment is small, but large enough that any given route covers clients in a bounded number of clusters.  This \emph{coarsens} the solutions considered by the algorithm, as vehicles must commit to covering larger fractions of load at a time. A dynamic program then finds the optimal such coarse solution using these simple building blocks within each cluster.

\subsection{Simplifying the Solution Structure}
Let $\hat{ep}$ and $\delta$ be problem-specific values that depend only on $\ep$. Let $\mathcal{H_T}$ denote the set of all subgraphs of $T$, and let $g:{\mathcal{H}}\rightarrow \mathbb{Z}^{\geq 0}$ be a problem-specific \emph{load function}.  We require $g$ to be monotonic and subadditive.  Intuitively, for all $H\in\mathcal{H_T}$, $g(H)$ is the load accumulated by a vehicle for covering $H$.

\subsubsection{Condensing the Input Tree}

The first step in the framework is to {\sc condense} all small branches into leaf edges. Specifically, let $\mathcal{B}$ be the set of all maximal branches of load at most $\delta D$.  That is, for every $v$-branch $b \in \mathcal{B}$, $g(b)  \leq \delta D$ and for $b$'s parent $p(v)$-branch, $b_p$, $g(b_p) >\delta D$.  For convenience, if $b_1 \in \mathcal{B}$ is a $v_1$ branch at $u$ and $b_2 \in \mathcal{B}$ is a \emph{sibling} $v_2$ branch at $u$ such that $g(b_1)+g(b_2) \leq \delta D$, we add a vertex $u'$ and an edge $(u,u')$ of length zero and replace $(u,v_i)$ with edge $(u',v_i)$ of length $\ell(u,v_i)$ for $i \in \{1,2\}$.  The $u'$ branch at $u$ then replaces the two branches $b_1$ and $b_2$ in $\mathcal{B}$. This ensures that any two branches in $\mathcal{B}$ with the same parent cannot be combined into a subtree of load $\leq \delta D$. 

Then, for every $b\in \mathcal{B}$, we \emph{condense} $b$ by replacing it with a leaf edge of length $\ell(b)$ and load $g(b)$.  All clients in $b$ are now assumed to be co-located at the leaf. Though it is easier to think of these condensed branches as leaf edges, the algorithm need not actually modify the input tree; condensing a branch is equivalent to requiring a single vehicle to cover the entire branch.  

\subsubsection{Clustering the Condensed Tree}

After condensing all small branches, we partition the condensed tree into clusters and define every leaf edge whose load is at least $\frac{\delta}{2} D$ to be a \emph{leaf cluster}.  The leaf-cluster-to-root paths define what we call the \emph{backbone} of $T$. By construction, every edge that is not on this backbone is either a leaf cluster (of load $\geq \frac{\delta}{2}D$) or a leaf edge (of load $< \frac{\delta}{2} D$).  That is, every vertex is at most one edge away from the backbone (see Figure~\ref{fig:woolly}). 

We can think of the condensed tree as a binary tree whose root is the depot, whose leaves are the leaf clusters, and whose internal vertices are the branching points of the backbone.  Each edge of this binary tree corresponds to a maximal path of the backbone between these vertices, together with the small leaf edges off of this path (see Figure~\ref{fig:woolly}).  To avoid confusion with tree edges, we call these path and leaf subgraphs \emph{woolly edges}.  A \emph{woolly subedge} of a woolly edge consists of a subpath of the backbone and all incident leaf edges.

\begin{figure}[!h]
\captionsetup[subfigure]{justification=centering}
\centering
  \subfloat[Woolly edges\label{fig:woolly}]
    {\includegraphics[width=0.3\textwidth]{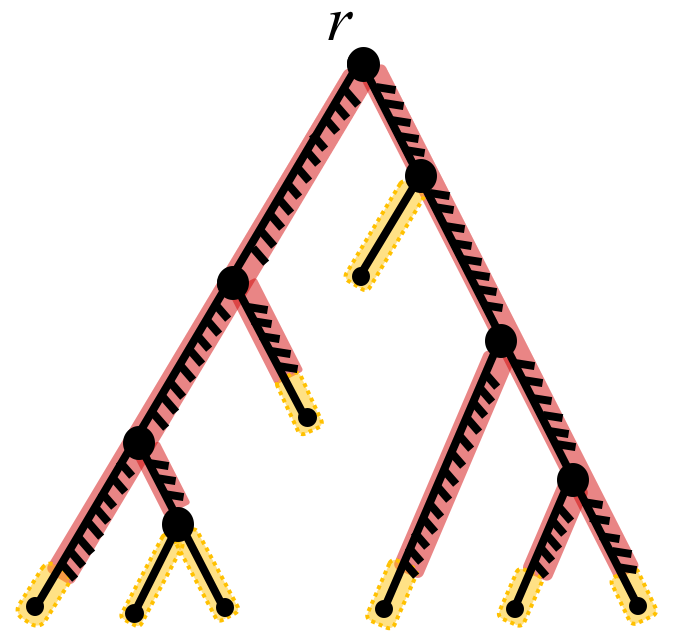}}
        \quad\vline\quad
  \subfloat[Clusters\label{fig:clusters}]
    {\includegraphics[width=0.3\textwidth]{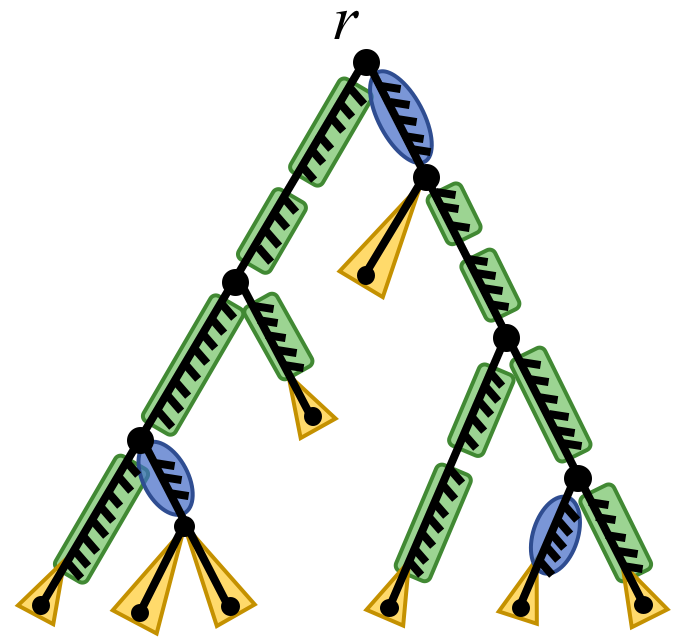}}
        \quad\vline\quad
  \subfloat[$T^*$\label{fig:t_star}]
    {\includegraphics[width=0.25\textwidth]{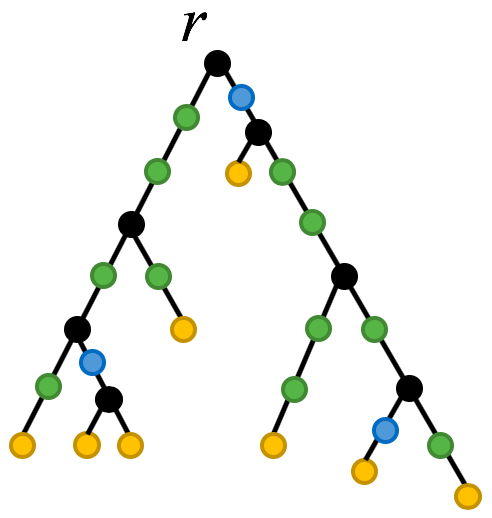}}
\caption{\label{fig:clustering} (a) Leaf clusters in yellow and woolly edges in red; (b) The tree partitioned into leaf clusters (yellow triangles), small clusters (blue ovals), and edge clusters (green rectangles); (c) The corresponding $T^*$ for clustering from (b).}
\end{figure}

A woolly edge $e$ whose load $g(e)$ is less than $\frac{\hat{\ep}\delta}{2}D$ is called a \emph{small cluster}.  The remaining woolly edges have load at least $\frac{\hat{\ep}\delta}{2}D$.  We partition each such woolly edge into one or more woolly subedges, which we call \emph{edge clusters}, each with load in $[\frac{\hat{\ep}\delta}{2}D,\frac{\delta}{2}D]$.   Backbone edges do not contain clients and can be subdivided as needed to ensure enough granularity in the tree edge lengths so that such a partition is always possible (see Figure~\ref{fig:clusters}).

For convenience, we label the components of edge clusters.  Let $\mathcal{C}$ be the set of edge clusters.  For any edge cluster $C \in \mathcal{C}$, let $P_C$ denote the backbone path in $C$ and let $L_C$ denote the leaf edges in $C$.  We order the backbone edges along $P_C$ as $p_{C,1}, p_{C,2},..., p_{C,m}$ in increasing distance from the depot and similarly label the leaf edges $e_{C,1}, e_{C,2}, ..., e_{C,m-1}$ such that $e_{C,i}$ is the leaf incident to $p_{C,i}$ and $p_{C,i+1}$ for all $1 \leq i < m$ (see Figure~\ref{fig:tour_types}).  If no such incident leaf exists for some $i$, we can add a leaf of length zero.  Likewise $P_C$ can be \emph{padded} with edges of length zero to ensure that each edge cluster `starts' and `ends' with a backbone edge.

\subsubsection{Solution Structure}

Consider the intersection of a solution with an edge cluster $C$. There are three different \emph{types} of routes that visit $C$ (see Figure~\ref{fig:tour_types}).  A $C$-\emph{passing} route traverses $C$ without covering any clients, and thus includes all of $P_C$ but no leaf edges in $L_C$. A $C$-\emph{collecting} route traverses \emph{and} covers clients in $C$, and thus includes all of $P_C$ and some edges in $L_C$. Last, a $C$-\emph{ending} route covers clients in, but does not traverse $C$, and thus includes backbone edges $p_{C,1}, p_{C,2},..., p_{C,i}$ for some $i<m$ and some leaves in $L_C$, but does not include all of $P_C$. Note that any $C$-\emph{ending} route can be assumed to cover some leaves in $L_C$ as removing any such redundancy would only improve a solution.

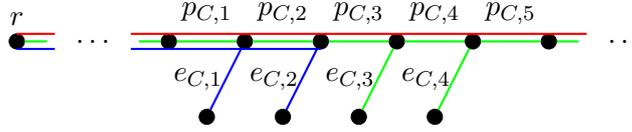
\begin{figure}
\begin{center}
\begin{tikzpicture}[every node/.style={circle,draw=black,minimum width=2mm,inner sep=0mm}]
\node[fill=black] (r) at (0,0) {};
\node[draw=none] (rlabel) at (0,0.3) {$r$};
\node[draw=none] (c1) at (0.5,0) {};
\node[draw=none] (c2) at (1.5,0) {};
\node[draw=none] at (1,0) {$\ldots$};
\node[fill=black] (v1) at (2,0) {};
\node[fill=black] (v2) at (3,0) {};
\node[fill=black] (v3) at (4,0) {};
\node[fill=black] (v4) at (5,0) {};
\node[fill=black] (v5) at (6,0) {};
\node[fill=black] (v6) at (7,0) {};
\node[draw=none] (c3) at (7.5,0) {};
\node[draw=none] at (8,0) {$\ldots$};
\node[fill=black] (l1) at (2.5,-1) {};
\node[fill=black] (l2) at (3.5,-1) {};
\node[fill=black] (l3) at (4.5,-1) {};
\node[fill=black] (l4) at (5.5,-1) {};
\draw[green,thick] (r) --  (c1);
\draw[red,thick] ($(r)+(0,0.1)$) --  ($(c1)+(0,0.1)$);
\draw[blue,thick] ($(r)+(0,-0.1)$) --  ($(c1)+(0,-0.1)$);
\draw[green,thick] (c2)--(v1);
\draw[red,thick] ($(c2)+(0,0.1)$) --  ($(v1)+(0,0.1)$);
\draw[blue,thick] ($(c2)+(0,-0.1)$) --  ($(v1)+(0,-0.1)$);
\draw[green,thick] (v1) -- node[black,draw=none,anchor=south] {$p_{C,1}$} (v2);
\draw[red,thick] ($(v1)+(0,0.1)$) -- ($(v2)+(0,0.1)$);
\draw[blue,thick] ($(v1)+(0,-0.1)$) -- ($(v2)+(0,-0.1)$);
\draw[green,thick] (v2) -- node[black,draw=none,anchor=south] {$p_{C,2}$} (v3);
\draw[red,thick] ($(v2)+(0,0.1)$) -- ($(v3)+(0,0.1)$);
\draw[blue,thick] ($(v2)+(0,-0.1)$) -- ($(v3)+(0,-0.1)$);
\draw[green,thick] (v3) -- node[black,draw=none,anchor=south] {$p_{C,3}$} (v4);
\draw[red,thick] ($(v3)+(0,0.1)$) -- ($(v4)+(0,0.1)$);
\draw[green,thick] (v4) -- node[black,draw=none,anchor=south] {$p_{C,4}$} (v5);
\draw[red,thick] ($(v4)+(0,0.1)$) -- ($(v5)+(0,0.1)$);
\draw[green,thick] (v5) -- node[black,draw=none,anchor=south] {$p_{C,5}$} (v6);
\draw[red,thick] ($(v5)+(0,0.1)$) -- ($(v6)+(0,0.1)$);
\draw[green,thick] (v6) --  (c3);
\draw[red,thick] ($(v6)+(0,0.1)$) --  ($(c3)+(0,0.1)$);
\draw[blue,thick] (v2) -- node[black,draw=none,anchor=east] {$e_{C,1}$} (l1);
\draw[blue,thick] (v3) -- node[black,draw=none,anchor=east] {$e_{C,2}$} (l2);
\draw[green,thick] (v4) -- node[black,draw=none,anchor=east] {$e_{C,3}$} (l3);
\draw[green,thick] (v5) -- node[black,draw=none,anchor=east] {$e_{C,4}$} (l4);
\end{tikzpicture}
\end{center}

\caption{Three types of route within an edge cluster $C$; the red tour is a $C$-passing route, the green tour is a $C$-collecting route, and the blue tour is a $C$-ending route.}
\label{fig:tour_types}
\end{figure}

We say that a cluster $C$ has \emph{single coverage} if a single vehicle covers \emph{all} clients in $C$.  We say that an edge cluster $C$ has \emph{split coverage} if there is one $C$-\emph{ending} route that covers leaf edges $e_{C,1}, e_{C,2}, ..., e_{C,i}$ for some $i<m-1$ and one $C$-\emph{collecting} route that covers leaf edges $e_{C,i+1}, e_{C,i+2}, ..., e_{C,m-1}$ (see Figure~\ref{fig:tour_types}). 

Finally, we say that a feasible solution has a \emph{simple structure} if:
\begin{itemize}
\item Leaf clusters and small clusters have single coverage,
\item Edge clusters have single or split coverage, and
\item Each vehicle covers clients in $O(\frac{1}{\hat{\ep}^2\delta})$ clusters
\end{itemize}

The framework relies on the proof of a \emph{structure theorem} stating that there exists a near-optimal solution (i.e. a feasible solution with maximum load at most $(1+\ep) D$)  with simple structure.  This proves that it is safe to reduce the set of potential solutions to those with simple structure.

\subsection{Dynamic Program}

After proving a structure theorem, the framework uses a dynamic programming algorithm (DP) to actually find a near-optimal solution with simple structure.  We define the \emph{cluster tree} $T^*$ to be the tree that results from contracting each cluster of $T$ to a single vertex. That is, the cluster tree has a vertex for each cluster and each branching point of the backbone (See Figure~\ref{fig:t_star}).  The DP traverses $T^*$ starting at the leaves and moving rootward, and enumerates the possible route structures within each cluster.  Namely, the DP considers all ways edge cluster coverage can be split and how routes are merged at branching points.

At each vertex in this tree the algorithm stores a set of \emph{configurations}.  A configuration is interpreted as a set of routes in $T$ that cover all clusters in the subtree of $T^*$ rooted at $v$.  Let $\theta\in(0,1]$ be a problem-specific value that depends only on $\ep$. A configuration at a vertex $v$ specifies, for each multiple $i$ of $\theta D$ between 0 and $(1+\ep)D$ the number of routes whose rounded load is $i$ at the time they reach $v$. Because $\theta$ depends only on $\ep$, the number of configurations and runtime of the DP is polynomially bounded.   After traversing the entire cluster tree, the solution is found at the root.  If there exists a configuration at the root such that all of the rounded route loads are at most $(1+\ep)D$, the algorithm returns this solution.

To ensure that the DP actually finds a near-optimal solution, we must bound the number of times that a given route is rounded to $\ep/\theta$, which gives a rounding error of at most $\ep D$.  In particular, we design the DP so that the number of times that any one route is rounded is proportional to the number of clusters that it covers clients in.  Then, using the structure theorem, there exists a near-optimal solution that covers clients in $O(\frac{1}{\hat{\ep}^2\delta})$ clusters and gets rounded by the DP $O(\frac{1}{\hat{\ep}^2\delta})$ many times.  Finally, $\theta$ is set to $c_{\theta}\ep\hat{\ep}^2\delta$ for some constant $c_{\theta}$.

For loads involving distance, $C$-passing routes pose a particular challenge for bounding rounding error.  These routes may accumulate load while passing through clusters without covering any clients, yet the DP cannot afford to update the load at every such cluster.  Instead, the DP \emph{projects} routes to predetermined destinations up the tree, so that they accumulate rounding error only once while passing many clusters. The configuration then stores the (rounded) loads of the \emph{projected} routes, and the DP need not update these load values for clusters passed through along the projection.

\subsection{Reassignment Lemma}\label{subsec:assignment_lemma}
We now present a lemma that will serve as a general-purpose tool for our framework. This tool is used to \emph{reassign} small route segments.  That is, if some subgraph $H$ is covered by several small route segments from distinct vehicles $h_1,h_2,...,h_m$, then for some $1\leq i \leq m$, the entire subgraph $H$ is assigned to be covered by $h_i$. This \emph{increases} load on $h_i$ so as to cover all of $H$, and \emph{decreases} load on $h_j$ for all $j\neq i$ which are no longer required to cover $H$ (see Figure~\ref{fig:condense}).  We show that this assignment process can be performed simultaneously for many such subgraphs such that the net load increase of any one route is small.

Let $G=(A,B,E)$ be an edge-weighted bipartite graph where $A$ is a set of \emph{facilities}, $B$ is a set of \emph{clients}, and $w(a,b) \geq 0$ is the weight of edge $(a,b) \in E$. For any vertex $v$, we use $N(v)$ to denote the \emph{neighborhood} of $v$, namely the set of vertices $u$ such that there is an edge  $(u,v)\in E$.  Each facility $a \in A$ has capacity $q(a) = \sum_{b\in N(a)}w(a,b)$ and each client $b\in B$ has weight $w(b) \leq \sum_{a\in N(b)}w(a,b)$.  A feasible \emph{assignment} is a function $f:B\rightarrow A$, such that each client $b$ is assigned to an adjacent facility $f(b)\in N(b)$. We can think of the weights $w(a,b)$ representing fractional assignment costs while weight $w(b)$ corresponds to a ``discounted'' cost of wholly serving client $b$. Ideally, the total weight of clients assigned to any facility $a$ would not exceed the capacity $q(a)$; however, this is not always possible.  We define the \emph{overload} $h_f(a)$ of a facility $a$ to be $w(f^{-1}(a))-q(a) = \sum_{b|f(b)=a}w(b)-\sum_{b\in N(a)}w(a,b)$ and the \emph{overload} $h_f$ of an assignment to be $\max_{a\in A}{h_f(a)}$. The {\sc Bipartite Weight-Capacitated Assignment} problem is to find an assignment with minimum overload.  

\begin{lemma}\label{lem:assign}
Given an instance of the {\sc Bipartite Weight-Capacitated Assignment} problem, an assignment with overload at most $\max_{b\in B}w(b)$ can be found efficiently.
\end{lemma}

\begin{proof}
Let $w_{max} = \max_{b\in B}w(b)$.  Consider an initial assignment $f$ by arbitrarily assigning each client to an adjacent facility.  Let $A_{0} = \{a\in A|h_f(a)>w_{max}\}$ be the set of facilities whose capacities are exceeded by more than $w_{max}$.  The lemma statement is satisfied if and only if $|A_{0}| = 0$. 

Let $B_0 = f^{-1}(A_{0})$ be the set of clients that are assigned to facilities in $A_{0}$.  We inductively define $A_i$ for $i \geq 1$ to be $N(B_{i-1})\setminus \bigcup_{j<i}A_j$ and $B_i= f^{-1}(A_i)$ to be the set of clients that are assigned to facilities in $A_i$ (see Figure~\ref{fig:assignment_lemma}). We say that a client $b$ (resp. facility $a$) has \emph{level} $i$ if $b\in B_i$ (resp. $a\in A_i$), and we say that client $b$ has infinite level if $b$ does not appear in any $B_i$.  By construction, each client appears in some $B_i$ for at most one value $i$, so the level of a client is either infinite or at most $|B|$ (see Figure~\ref{fig:assignment_lemma}).

\begin{figure}
\begin{center}
\begin{tikzpicture}[scale=0.8,every node/.style={circle,draw=black,inner sep=0.5mm}]
\node[draw=none] at (-2,0) {$A$};
\node[draw=none] at (-2,-2) {$B$};
\node[] (a1) at (0,0) {$a_1$};
\node[] (a2) at (1,0) {$a_2$};
\node[] (a3) at (2,0) {$a_3$};
\node[] (a4) at (3,0) {$a_4$};
\node[] (b1) at (-1,-2) {$b_1$};
\node[] (b2) at (0,-2) {$b_2$};
\node[] (b3) at (1,-2) {$b_3$};
\node[] (b4) at (2,-2) {$b_4$};
\node[] (b5) at (3,-2) {$b_5$};
\node[] (b6) at (4,-2) {$b_6$};
\draw[thick] (a1)--(b1);
\draw[thick] (a1)--(b2);
\draw[thick] (a2)--(b3);
\draw[thick] (a2)--(b4);
\draw[thick] (a3)--(b5);
\draw[thick] (a4)--(b6);
\draw[thick,dashed] (a2)--(b2);
\draw[thick,dashed] (a3)--(b1);
\draw[thick,dashed] (a3)--(b3);
\draw[thick,dashed] (a3)--(b4);
\draw[thick,dashed] (a3)--(b6);
\draw[thick,red] (-0.5,0.5) rectangle (1.45,-0.5);
\node[red,draw=none] at (0.5,0.65) {$A_0$};
\draw[thick,red] (-1.5,-1.5) rectangle (2.45,-2.5);
\node[red,draw=none] at (0.5,-2.75) {$B_0$};
\draw[thick,blue] (1.55,0.5) rectangle (2.5,-0.5);
\node[blue,draw=none] at (2,0.65) {$A_1$};
\draw[thick,blue] (2.55,-1.5) rectangle (3.5,-2.5);
\node[blue,draw=none] at (3,-2.75) {$B_1$};
\end{tikzpicture}
\end{center}
\caption{An example of the inductive subsets $A_i$ and $B_i$. A solid edge $(a,b)$ indicates that $f(b) = a$, and a dashed edge indicates that $(a,b) \in E$ but $f(b) \neq a$. Note that $b_6$'s level is infinite.}
\label{fig:assignment_lemma}
\end{figure}
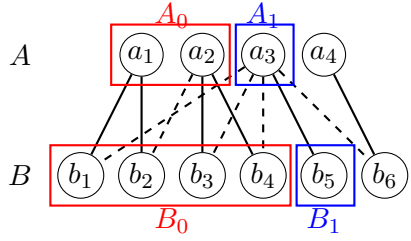

It suffices to show that if $|A_{0}| > 0$ then there is some client $b$, with some finite level, whose level can be increased without decreasing the level of any other client.  After at most $|B|^2$ such improvements, $B_0$ must be empty, so $|A_{0}|$ must also be empty, proving the claim.

Suppose that there is some $i$ and some facility $a \in A_i$ such that $h_f(a)\leq q(a)$.  We say such a facility is \emph{underloaded}, and therefore $i>0$.  By construction, there is some $b \in B_{i-1}$ adjacent to $a$ and some $a' \in A_{i-1}$ such that $f(b) = a'$.  We reassign $b$ to $a$ by setting $f(b) = a$.  Note that since this adds $w(b) \leq w_{max}$ to the load of $a$, the resulting overload of $a$ is at most $w_{max}$, so the level of $a$ does not decrease.  Further, $b$ now has level at least $i$ and the level of every other client is either unaffected or increased.


We now show that such an underloaded facility always exists. Let $j$ be the largest value such that $A_j$ is non-empty.  If $B_j$ is empty, then all facilities in $A_j$ are underloaded.  Otherwise $B_j$ is non-empty and $N(\bigcup_{0\leq i \leq j} B_i) \subseteq \bigcup_{0\leq i \leq j} A_i$, so $\sum_{b\in \cup_iB_i}w(b) \leq \sum_{b\in \cup_iB_i}\sum_{a\in N(b)}w(a,b) \leq \sum_{a\in \cup_iA_i}q(a)$. Since no other clients are assigned to these facilities, at least one facility in $\bigcup_{0\leq i \leq j} A_i$ must be underloaded.

\end{proof}
\section{Customizing the Framework: {\sc Minimum Makespan Vehicle Routing}}\label{sec:makespan}

In this section, we demonstrate how to apply the general framework to a specific problem, {\sc Minimum Makespan Vehicle Routing}.  In particular, we use the framework to achieve the following:

\begin{theorem}\label{thm:makespan_ptas}
For every $\ep > 0$, there is a polynomial-time algorithm that, given an instance of {\sc Minimum Makespan Vehicle Routing} on a tree, finds a solution whose makespan is at most $1+\ep$ times optimum.
\end{theorem}

Recall that the problem is to find $k$ \emph{tours} that serve all clients in $T$ such that the maximum \emph{length} of any tour is minimized. The vehicle routes are tours, and the vehicle load is tour length, so the load $g(H)$ of subgraph $H$ is twice the length of edges in the subgraph.  The {\sc condense} operation is then applied to the input tree, with $\delta= \hat{\ep} = \ep/c$ for some constant $c$ we will define later.  Leaf clusters therefore correspond to branches of length at least $\frac{\hat{\ep}}{4}D$ (load at least $\frac{\hat{\ep}}{2}D$), small clusters have total length less than $\frac{\ep\hat{\ep}}{4}D$, and edge clusters have total length in $[\frac{\ep\hat{\ep}}{4}D,\frac{\hat{\ep}}{4}D]$.  As described in Section~\ref{sec:framework}, the two steps in applying the framework are proving a structure theorem and designing a dynamic program.

\subsection{{\sc Minimum Makespan Vehicle Routing} Structure Theorem}

We prove the following for {\sc Minimum Makespan Vehicle Routing}.

\begin{theorem}\label{thm:structure_makespan}
If there exists a solution with makespan $D$, then there exists a solution with makespan at most $1+O(\hat{\ep})D$ that has simple structure.
\end{theorem}

We prove the above by starting with some optimal solution of makespan at most $D$ and show that after a series of steps that transforms the solution into one with simple structure, the makespan is still near-optimal.  

To ensure that each step maintains solution feasibility, we introduce the following notion of independence.  Let $T'$ be a connected subgraph of $T$ containing the depot $r$, and let $X$ be a set of subgraphs of $T$. We say that $X$ is a \emph{tour-independent} set with respect to $T'$ if $T' \cup X'$ is connected for all $X' \subseteq X$. In particular, if $T'$ is the subgraph covered by a single tour then adding any subgraphs in $X'$ creates a new feasible tour.  

\begin{lemma}\label{lem:condense}
The {\sc condense} operation adds at most $\hat{\ep} D$ to the optimal makespan.
\end{lemma}
\begin{proof}
The {\sc condense} operation is equivalent to requiring every branch in $\mathcal{B}$ to be covered by a single tour.  We show that there is such a solution of makespan at most $OPT + \hat{\ep} D$. Fix an optimal solution, and let $A$ be the set of tours in the optimal solution that (at least partially) cover branches in $\mathcal{B}$.  We define an edge-weighted bipartite graph $G =(A,\mathcal{B},E)$ where there is an edge $(a,b)$ if and only if tour $a$ contains edges of branch $b$, and $w(a,b)$ is the length of the tour segment of $a$ in branch $b$, namely twice the length of the edges covered by tour $a$. Note that $\forall a\in A,b\in \mathcal{B}$, $w(a,b)\leq \hat{\ep} D$.  For each $b\in \mathcal{B}$, we define the weight $w(b)$ to be $2\ell(b)$, and for each $a \in A$, we define the capacity $q(a)$ to be the sum $\sum_{b:a\cap b \neq \emptyset} w(a,b)$ of all tour segments of $a$ in branches of $\mathcal{B}$. Clearly, $w(b) \leq \sum_{a:a\cap b \neq \emptyset} w(a,b)$, since these tour segments collectively cover $b$.  

Essentially, $q(a)$ represents tour $a$'s \emph{budget} for  buying whole branches and is defined by the length of its tour segments in the branches that it partially covers. Further, we will only assign a branch to a tour that already covers some edges in the branch so there is no additional cost to connect the tour to the branch. 

Applying Lemma~\ref{lem:assign} to $G$, we can achieve an assignment of branches to tours such that each branch is assigned to one tour and the capacity of each tour is exceeded by at most $\max_{b\in B} w(b) \leq \hat{\ep} D$. Further, for any tour $a \in A$, let $T'_a$ be the corresponding subgraph visited by $a$ excluding any branches in $B$. $T'_a$ contains $r$ and is connected, so $N_G(a)\subseteq \mathcal{B}$ is a tour-independent set with respect to $T'_a$. Thus, the reassignment of branches creates a feasible solution in which the extra distance traveled by each tour is at most $\hat{\ep} D$. 
\end{proof}

\begin{figure}[!h]
\captionsetup[subfigure]{justification=centering}
\centering
  \subfloat[\label{fig:small_segments}]
    {\includegraphics[width=0.135\textwidth]{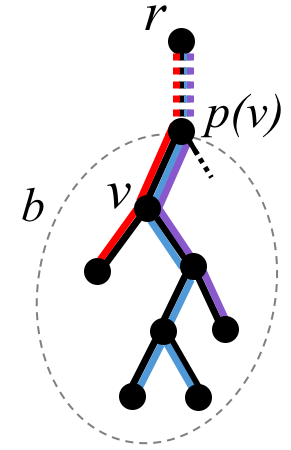}}
        \quad\vline\quad
  \subfloat[\label{fig:assign}]
    {\includegraphics[width=0.1\textwidth]{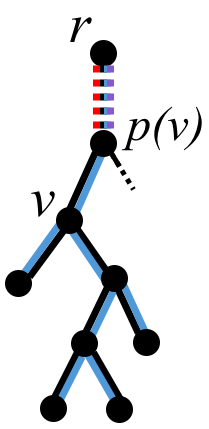}}
    \quad\vline\quad
  \subfloat[\label{fig:condensed}]
    {\includegraphics[width=0.1\textwidth]{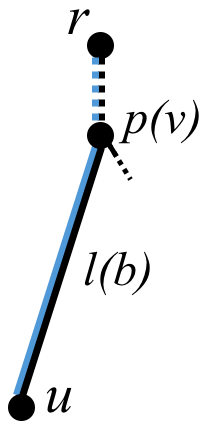}}
\caption{\label{fig:condense} Figure (a) depicts a branch $b \in \mathcal{B}$ covered by several small tour segments; (b) shows the entire branch $b$ being assigned to the blue tour; (c) shows the result of the condense operation.}
\end{figure}

\begin{lemma}\label{lem:small}
Requiring all leaf clusters and small clusters to have single coverage increases the makespan by at most $4\hat{\ep}D$.
\end{lemma}
\begin{proof}
After condensing the tree, all leaf clusters have single coverage, and the effect on makespan was covered in Lemma~\ref{lem:condense}. Because of the binary tree structure, we can \emph{assign} each small cluster to a descendant leaf cluster in such a way that each leaf cluster is assigned at most two small clusters.  Since each leaf cluster is covered by a single tour, we can require this tour to also cover the clients of the small cluster(s) assigned to that leaf cluster.  This is feasible since small clusters are only assigned to \emph{descendant} leaf clusters.  Furthermore, since leaf clusters have length at least $\frac{\hat{\ep}}{4}D$, we can \emph{charge} this error to the length of the leaf clusters. In particular, since any given tour covers at most $D/(2\cdot\frac{\hat{\ep}}{4} D) = \frac{2}{\hat{\ep}}$ leaf clusters, this assignment adds at most $2\cdot \frac{2}{\hat{\ep}} \cdot (2\cdot \frac{\ep\hat{\ep}}{2}D) = 4\hat{\ep} D$ to the makespan.  
\end{proof}

\begin{lemma}\label{lem:split}
Requiring every edge cluster to have single or split coverage adds at most $3\hat{\ep} D$ to the optimal makespan.
\end{lemma}

\begin{proof}
Fix an optimal solution, $SOL_0$.  We begin by limiting the number of $C$-ending tours. For every edge cluster $C$, let $\gamma(C)$ be the subcluster of $C$ covered by $C$-ending tour segments.  Let $\mathcal{C}_1 =\{\gamma(C)|C \in \mathcal{C}\}$ and let $A_1$ be the set of tours in $SOL_0$ that are $C$-ending tours for at least one edge cluster $C \in \mathcal{C}$.  
We define a bipartite graph $G_1 =(A_1,\mathcal{C}_1,E_1)$ where there is an edge $(a,\gamma(C))$ if and only if tour $a$ is a $C$-ending tour, and $w_{G_1}(a,\gamma(C))$ is the length of the tour segment of $a$ in edge cluster $C$. Note that $\forall a\in A_1,\gamma(C)\in \mathcal{C}_1$, $w_{G_1}(a,\gamma(C))\leq 2\ell(C) \leq\hat{\ep} D$.  We define the weight $w_{G_1}(\gamma(C)) = 2 \ell(\gamma(C))$, and for each $a \in A$, we define the capacity $q(a)$ to be the sum $\sum_{C:\gamma(C) \in N_{G_1}(a)} w_{G_1}(a,\gamma(C))$.  Clearly, $w_{G_1}(\gamma(C)) \leq \sum_{a \in N_{G_1}(\gamma(C))} w_{G_1}(a,\gamma(C))$, since these tour segments collectively cover $\gamma(C)$.  Therefore we can apply Lemma~\ref{lem:assign} to $G_1$ to achieve an assignment of  subclusters in $\mathcal{C}_1$ to tours such that each subcluster is assigned to one tour and the capacity of each tour is exceeded by at most $\max_{\gamma(C)\in \mathcal{C}_1} w_{G_1}(\gamma(C)) \leq \hat{\ep} D$. For any tour $a \in A_1$, let $T'_a$ be the corresponding subgraph visited by $a$, excluding any edge clusters $C$ for which $a$ is a $C$-ending tour.  $T'_a$ contains $r$ and is connected. Since, for each tour $a$, $N_{G_1}(a)$ is a tour-independent set with respect to $T'_a$, this assignment forms a new feasible solution, $SOL_1$.

At this point each edge cluster $C$ has at most one $C$-ending tour.  We now address $C$-collecting tours.  For every edge cluster $C$, let $\gamma(L_C)\subseteq L_C$ denote the set of leaf edges of $C$ that are covered by $C$-collecting tours.  Let $\mathcal{L}_2 =\{\gamma(L_C)| C \in \mathcal{C}\}$, and let $A_2$ be the set of tours in $SOL_1$ that are $C$-collecting tours for at least one cluster $C \in \mathcal{C}$.  We define a bipartite graph $G_2 =(A_2,\mathcal{L}_2,E_2)$ where there is an edge $(a,\gamma(L_C))$ if and only if tour $a$ is a $C$-collecting tour, and $w_{G_2}(a,\gamma(L_C))$ is twice the length of the leaves in $\gamma(L_C)$ covered by $a$.  Note that $\forall a\in A_2,C\in \mathcal{C}$, $w_{G_2}(a,\gamma(L_C))\leq 2 \ell(C) \leq\hat{\ep} D$.  We define the weight $w_{G_2}(\gamma(L_C)) = 2\ell(\gamma(L_C))$, and for each $a \in A$, we define the capacity $q(a)$ to be the sum $\sum_{C:\gamma(C) \in N_{G_2}(a)} w_{G_2}(a,\gamma(L_C))$ of tour segments covering leaves in $C$-collecting tour segments of $a$.  Clearly, $w_{G_2}(\gamma(L_C)) \leq \sum_{a\in N_{G_2}(\gamma(C))} w_{G_2}(a,C)$, since these tour segments collectively cover $\gamma(L_C)$.  Therefore we can apply Lemma~\ref{lem:assign} to $G_2$ to achieve an assignment of $\mathcal{L}_2$ leaf sets to tours such that each leaf set is assigned to a single tour and the capacity of each tour is exceeded by at most $ \max_{C\in \mathcal{C}} w_{G_2}(\gamma(L_C)) \leq \hat{\ep} D$. 
For any tour $a \in A_2$, let $T''_a$ be the corresponding subgraph visited by $a$, excluding all leaf sets in $L_C$.  $T''_a$ contains $r$ and is connected. 
Since, for each tour $a$, $N_{G_2}(a)$ is a tour-independent set with respect to $T''_a$, this assignment forms a new feasible solution, $SOL_2$.  Note that if some $C$-collecting tour $a'$ is not assigned $\gamma(L_C)$, it becomes a $C$-passing tour.

At this point every edge cluster $C$ has at most one $C$-collecting tour and at most one $C$-ending tour. If $C$ has single or split coverage, we are done. Otherwise, let $a_1$ be the $C$-ending tour and $a_2$ be the $C$-collecting tour.  Let $j<m$ be the largest index such that $a_1$ covers leaf $e_{C,j}$.  Since $j$ is the largest such index, then $a_2$ covers all leaves $e_{C,i}$ for $j<i<m$.  If $a_1$ covers \emph{all} leaves $e_{C,1},e_{C,2},...,e_{C,j}$, then $C$ has split coverage, and we are done.  Otherwise, both $a_1$ and $a_2$ contain the backbone edges $P_{C}' = \{p_{C,1}, \ldots, p_{C,j}\}$ and a subset of the leaf edges $L_{C}' = \{e_{C,1}, \ldots, e_{C,j}\}$. Using the same argument as above for $C$-collecting tours, we can assign the leaves $L_{C}'$  to exactly one of these tours while adding at most $\max_{C\in \mathcal{C}} 2\ell(L_{C}') \leq \hat{\ep} D$. If the leaves are assigned to $a_1$, then $C$ ends up having split coverage, and if the leaves are assigned to $a_2$, then the remaining tour segment of $a_1$ along the backbone is redundant and can be removed, resulting in every edge cluster $C$ having single coverage.  

The resulting solution $SOL_3$ satisfies the lemma statement.
\end{proof}

We now prove Theorem~\ref{thm:structure_makespan}.

\begin{proof}
Given Lemma~\ref{lem:condense}, Lemma~\ref{lem:small}, and Lemma~\ref{lem:split}, all that remains in proving Theorem~\ref{thm:structure_makespan} is to bound the number of clusters that a single vehicle covers clients in.   To do so, we first show that requiring the length of every $C$-ending or $C$-collecting tour segment to be at least $\frac{\hat{\ep}^3}{2}D$ adds a stretch factor of at most $1+\hat{\ep}$ to the makespan.

Consider any edge cluster $C$ with split coverage.  Let $a_0$ be the $C$-ending tour, and let $a_1$ be the $C$-collecting tour. For $i = 0,1$ let $L_i$ be the length of $a_i$ within $C$. Suppose that some $L_i$ is less than $\frac{\hat{\ep}^3}{2}D$.  Since $L_{0}+L_{1} \geq 2\ell(C) \geq 2\frac{\hat{\ep}^2}{2}D = \hat{\ep}^2D$, then $L_{1-i} \geq \hat{\ep}^2(1-\frac{\hat{\ep}}{2})D$ and assigning $a_{1-i}$ to cover the clients served by $a_i$ increases the length of $a_{1-i}$ within $C$ by at most a factor of $1+\hat{\ep}$ since $L_{i}/L_{i-1} \leq \hat{\ep}$.  Reassigning these small tour segments within all edge clusters with split coverage increases the length of any given tour by at most a factor of $1+\hat{\ep}$.

After all of the above steps, there exists a solution with makespan at most $(1+\hat{\ep})(1+8\hat{\ep})D$. Consider some such solution and a single tour $a$.  Let $X_1$ be the set of leaf clusters that $a$ covers.  Since any such leaf cluster has length at least $\frac{\hat{\ep}}{4}D$ and is entirely covered by $a$, $|X_1|$ is $O(\frac{1}{\hat{\ep}})$.  Let $X_2$ be the set of small clusters that $a$ covers.  Since $a$ covers exactly those small clusters that are assigned to leaf clusters in $X_1$, $|X_2|\leq 2|X_1|$ and is also $O(\frac{1}{\hat{\ep}})$, since each leaf cluster is assigned at most two small clusters. Last, let $X_3$ be the set of edge clusters $C$ for which $a$ is either a $C$-ending tour or a $C$-collecting tour.  Since the length of tour $a$ is at most $(1+O(\hat{\ep}))D$, and each of $t$'s tour segments in $X_3$ has length at least $\frac{\hat{\ep}^3}{2}D$, $|X_3|$ is $O(\frac{1}{\hat{\ep}^3})$.

\end{proof}

\subsection{{\sc Minimum Makespan Vehicle Routing} Dynamic Program}\label{sec:DP}

Having proven a structure theorem, we now present a dynamic programming algorithm (DP) that actually finds a near-optimal solution with simple structure.

Recall, the DP traverses cluster tree $T^*$ starting at the leaves and moving rootward.  A configuration is a vector in $\{0,1,2,...,k\}^{2\hat{\ep}^{-4}}$.  A configuration $\vec{x}$ at a vertex $v$ is interpreted as a set of tours \emph{projected up to $r$} in $T$ that cover all clusters in the subtree of $T^*$ rooted at $v$. For $i \in \{1,2,...,2\hat{\ep}^{-4}\}$, $\vec{x}[i]$ is the number of tours in the set that have \emph{rounded} length $i\hat{\ep}^{4}D$.  That is, the actual tours that correspond to the $\vec{x}[i]$ tours represented in the vector each have length that may be less than $i\hat{\ep}^{4}D$. 

The algorithm categorizes the vertices into three different cases and handles them separately.  The \emph{base cases} are the leaves of $T^*$.  Let $v\in T^*$ be such a leaf, let $L_v$ be the corresponding leaf cluster in $T$, and let $u$ be the vertex at which $L_v$ meets the backbone.  When the algorithm determines the configuration for $v$ it addresses covering both $L_v$ as well as covering any small clusters $C_1,...,C_h$ that are assigned to $L_v$.  Let $\ell_{small}$ be the length of all of the \emph{leaves} of these small clusters, namely $\ell_{small} = \ell(\bigcup_{1\leq i\leq h}C_i\setminus\text{backbone})$.  Let $\ell_0$ be $2(\ell(L_v) + \ell_{small} + d_T(u,r))$ rounded up to the nearest $\hat{\ep}^4D$.  The only configuration stored at $v$ is $\vec{x}$ such that $\vec{x}[\ell_0] = 1$ and $\vec{x}[j]=0, \forall j\neq \ell_0$.  All cluster lengths and distances to the depot can be precomputed in linear time, after which each base case can be computed in constant time.

The \emph{grow cases} are the vertices in $T^*$ that correspond to edge clusters in $T$.  Let $v\in T^*$ be such a vertex, and let $C_v$ be the corresponding edge cluster in $T$.  Let $u$ be the root-most vertex in $C_v$, and let $v'\in T^*$ be the lone child vertex of $v$.  Note that $v'$ may correspond to a branching backbone vertex, a leaf cluster or another edge cluster, but by construction, $v$ has exactly one child.  Since $C_v$ has single or split coverage, at most two tours in any configuration at $v$ are involved in covering the leaves of $C_v$: all other tours in the configuration are $C_v$-passing tours, and their representation in the configuration remains unchanged.  The algorithm considers all possible rounded tour lengths $\ell_1$ for a $C_v$-ending tour $t_1$ for the configuration (including not having such a tour) and for each such $t_1$, the algorithm considers all possible (rounded) lengths $\ell_2$ for an incoming $C_v$-collecting tour $t_2$, \emph{before} the remaining length from covering leaves in $C_v$ is added to the tour.  Given $\ell_1$ and $\ell_2$, the algorithm can easily compute the resulting rounded length $\ell_3$ of $t_2$ \emph{after} covering its share of $C_v$ leaves.  For each configuration $\vec{x}'$ for child vertex $v'$, the algorithm determines configuration $\vec{x}$ for $v$ such that $\vec{x}[\ell_1]=\vec{x}'[\ell_1]+1$, $\vec{x}[\ell_2]=\vec{x}'[\ell_2]-1$, $\vec{x}[\ell_3]=\vec{x}'[\ell_3]+1$, and $\vec{x}[i]=\vec{x}'[i]$ otherwise.  If the resulting $\vec{x}$ is feasible, it is stored at $v$.  Since there are at most $2\hat{\ep}^{-4}$ options for $\ell_1$ and $\ell_2$ and at most $k^{2\hat{\ep}^{-4}}$ configurations at $v'$, the runtime for each grow case is $k^{O(\hat{\ep}^{-4})}$.

Finally, the \emph{merge cases} are the vertices in $T^*$ that correspond to branching backbone vertices in $T$ as well as the depot.  Let $v\in T^*$ be such a vertex, and let $u$ be the corresponding vertex in $T$. Let $v_1, v_2\in T^*$ be the two children of $v$ in $T^*$.  Every tour $t$ in a configuration at $v$ will either be a directly inherited tour $t_i$ of rounded length $\ell_i$ from a configuration at $v_i$ for $i \in \{1,2\}$, or will be a \emph{merging} of some tour $t_1$ from $v_1$ and some $t_2$ from $v_2$ with resulting length $\ell_1 + \ell_2 -2\ell(u,r)$ rounded up to the nearest $\hat{\ep}^4D$ (recall that $t_1$ and $t_2$ are tours from the depot so the subtracted amount addresses over-counting the path to the depot).  For every possible $(\ell_1, \ell_2)$ (including lengths of zero to account for tours inherited by children), the algorithm \emph{guesses} the number of tours at $v$ that resulted from merging a tour of length $\ell_1$ from $v_1$ with a tour of length $\ell_2$ from $v_2$.  Each of these possibilities corresponds to a configuration $\vec{x}_i$ at $v_i$ for $i\in \{1,2\}$ and to a merged configuration $\vec{x}$ at $v$.  If $\vec{x}_1$ and $\vec{x}_2$ are valid configurations stored at $v_1$ and $v_2$, respectively, then the algorithm stores $\vec{x}$ at $v$.  There are $k^{4\hat{\ep}^{-8}}$ such possibilities, so the runtime of each merge case is $k^{O(\hat{\ep}^{-8})}$. Note that the dynamic program only considers storing feasible configurations $\vec{x}$ at vertex $v$ so the algorithm maintains that there are at most $k$ tours total. 

Since for any $\ep>0$ the DP has a polynomial runtime, the following lemma completes the proof of Theorem~\ref{thm:makespan_ptas}.

\begin{lemma}\label{lem:dp} The dynamic program described above finds a tour with maximum makespan at most $(1+\ep) D$.
\end{lemma}

\begin{proof}
By Theorem~\ref{thm:structure_makespan}, each tour covers clients in $O(1/\hat{\ep}^3)$ clusters and has accumulated $c_1\hat{\ep}D$ error during the simplification steps, for some constant $c_1$.  Notably, each tour only covers the backbone edges that lead to cluster in which it covers clients, since otherwise these edges can be removed.  Therefore, the number of backbone branching vertices that a tour branches at is also $O(1/\hat{\ep}^3)$.  Note that the tour might pass other backbone branching vertices without the tour itself actually branching there.  The dynamic program rounds tour lengths up to the nearest $\hat{\ep}^{4}D$ exactly at the clusters in which the tour covers clients and at the backbone branching vertices at which the tour branches.  Therefore each tour incurs an error of at most $\hat{\ep}^{4}D$ at each of the $O(1/\hat{\ep}^3)$ times that the dynamic program rounds the tour, for a total rounding error of $c_2\hat{\ep} D$ for some constant $c_2$.  The total increase in the makespan is $(c_1+c_2)\hat{\ep}D$.  Setting $\hat{\ep}$ to $\frac{\ep}{c_1+c_2}$ completes the proof.
\end{proof}

\section{{\sc Capacitated Vehicle Routing}}\label{sec:capacity}

Recall that in {\sc Capacitated Vehicle Routing}, each tour can cover at most $Q$ clients and the objective is to minimize the \emph{sum} of the tour lengths (there is no constraint on the number of tours or the length of each tour).  Because {\sc Capacitated Vehicle Routing} is a {\sc Minimum Fleet Budget} problem, we must apply the framework directly to the associated {\sc Min-Max Vehicle Load} problem, {\sc Min-Max Client Capacity} that seeks to minimize the maximum client capacity of a vehicle, $Q$, given a bound $k$ on the sum of tour lengths.

Customizing the framework to {\sc Min-Max Client Capacity} is nearly identical to {\sc Minimum Makespan Vehicle Routing}, except that vehicle load is the number of clients, rather than length, and the fleet budget is sum of tour lengths, rather than the number of vehicles.  Correspondingly, we define the load $g(H)$ of a subgraph $H$ to be the \emph{number} of clients in the subgraph.  Otherwise, the solution simplification and analysis proceed as in {\sc Minimum Makespan Vehicle Routing}, noting that because the objective function is the sum of tour lengths, assigning an entire branch to be covered by a single tour never increases the fleet budget.

The {\sc Min-Max Client Capacity} DP configuration is a vector in $\{0,1,2,...,n\}^{\hat{\ep}^{-4}}$, where a configuration $\vec{x}$ at a vertex $v$ is interpreted as a set of tours in $T$ that cover all clusters in the subtree of $T^*$ rooted at $v$.  Now, however, for $i \in \{1,2,...,\hat{\ep}^{-4}\}$, $\vec{x}[i]$ is the number of tours in the set that have rounded \emph{number of clients} $i\hat{\ep}^{4}Q$.  That is, the actual tours that correspond to the $\vec{x}[i]$ tours represented in the vector each cover at most $i\hat{\ep}^{4}Q$ clients. For each of these configurations, the DP stores the minimum total length of a solution that is consistent with the configuration.

The resulting PTAS for {\sc Min-Max Client Capacity} gives the following bicriteria result for {\sc Capacitated Vehicle Routing} as a corollary.

\begin{theorem}
Given an instance of {\sc Capacitated Vehicle Routing} on a tree, if there exists a solution of total length $k$ and capacity $Q$, then for any $\ep > 0$, there is a polynomial-time algorithm that finds a solution of total length $k$ and capacity at most $(1+\ep)Q$.
\end{theorem}

\section{School Bus Routing}\label{sec:school_bus}

A rather different vehicle load measure is maximum client \emph{regret}: the maximum difference between the time a client is \emph{actually} visited and the earliest possible time that client \emph{could have} been visited.  Note that maximum client regret is always achieved for the \emph{final} client visited by a vehicle.  It is therefore convenient to assume the routes to be \emph{paths}. Recall, the regret of a $u$-to-$v$ path, $P$, is the difference between the length of the path and the distance from $u$ to $v$, namely $l(P)-d_T(u,v)$.

The {\sc Min-Max Regret Routing} problem, given a number of vehicles, $k$, is to cover all clients with $k$ paths starting from the depot, such that the maximum regret of a path is minimized. Here, vehicle load is the regret of the vehicle's path and fleet budget is the number of vehicles.  The related {\sc Minimum Fleet Budget} problem, the {\sc School Bus Routing} problem, given a regret bound $R$, is to find the smallest number of paths of regret at most $R$ that each start at the depot, $r$, and collectively cover all clients.  

In this section, we sketch how to customize our framework to the {\sc Min-Max Regret Routing} and achieve the following.

\begin{theorem}
\label{thm:regret}
For every $\ep > 0$, there is a polynomial-time algorithm that, given an instance of {\sc Min-Max Regret Routing} on a tree, finds a solution whose maximum regret is at most $1+\ep$ times optimum.
\end{theorem}

Both a bicriteria approximation scheme and a 2-approximation for the {\sc School Bus Routing} problem follow as corollaries.

\begin{theorem}\label{thm:school_bus}
Given an instance of the {\sc School Bus Routing} problem on a tree, if there exists a solution consisting of $k$ paths of regret at most $R$, then for any $\ep > 0$, there is polynomial-time algorithm that finds a solution consisting of $k$ paths of regret at most $(1+\ep)R$.
\end{theorem}

\begin{theorem}
There is a polynomial-time 2-approximation for the {\sc School Bus Routing} problem in trees.
\end{theorem}
\begin{proof}
\cite{friggstad14} show that $k$ paths of regret at most $\alpha R$ can be replaced with at most $\lceil\alpha\rceil k$ paths each of regret at most $R$.  Applying this to Theorem~\ref{thm:school_bus} completes the proof.
\end{proof}

The ability of our framework to capture a regret-based load function demonstrates the flexibility of the framework and gives evidence for its potential application to problems outside the scope of this paper. 

It is useful (and equivalent) to assume that each vehicle path \emph{ends} at the depot rather than begins there; we assume each vehicle starts at some origin vertex and travels to the depot $r$, possibly taking some detours on branches along the way.  Note that in tree instances, these \emph{regret detours} exactly correspond to the regret accumulated by the vehicle, whereas traveling along the origin-to-depot path is free.  The biggest challenge with applying the framework to {\sc Min-Max Regret Routing} is the asymmetry of these two aspects of the path.  

\subsection{Structure Theorem}

For {\sc Min-Max Regret Routing}, defining a load function for subgraphs is not as straightforward as for other problems, as the regret accumulated by a path covering subgraph $H$ is highly dependent on where it starts and ends.  If the vehicle is traversing $H$ on its way to the depot, then some of $H$ can be covered for free, namely the intersection of $H$ with the shortest path from the vehicle's origin to the depot.  For this reason, we require $g(\cdot)$ to take two additional parameters, vertices $u,v\in H$, and define $g(H,u,v)$ to be the minimum regret required for a path to cover $H$, given that $u$ and $v$ are respectively the \emph{first} and \emph{last} vertex in $H$ on the path.  This means that $u$ and $v$ are each either end points of the path or on the \emph{boundary} of $H$.

Let $\hat{\ep} = \ep/c$ for some constant $c$ we will define later, and $\delta=\hat{\ep}^2$. We first update the condensing and clustering steps to use this extended load function. For a branch $b$ that is a $u$-branch at $v$, the branch load is measured as $g(b,v,v)$ when applying the {\sc condense} operation.  That is, the regret accumulated by a vehicle covering $b$ and starting \emph{outside of} $b$, or equivalently, twice the length of $b$.  Leaf clusters therefore correspond to branches of length at least $\frac{\hat{\ep}^2}{4}D$ (i.e. load at least $\frac{\hat{\ep}^2}{2}D$).  Let the leaf edges of a woolly subedge be called the \emph{wool}.  Every woolly subedge can therefore be partitioned into backbone and wool.  The load of a woolly subedge $e = (u,v)$ is defined as $g(e,u,v)$, which is twice the length of the wool of $e$ (i.e. twice the sum of the leaf edge lengths).  Therefore, small clusters have total wool length less than $\frac{\hat{\ep}^3}{4}D$ and edge clusters have total wool length in $[\frac{\hat{\ep}^3}{4}D,\frac{\hat{\ep}^2}{4}D]$

The condense operation increases the regret of any path by at most $2\hat{\ep}^2D$.  If any path starts in a branch $b$ that is condensed, it can be assumed instead that the branch starts at the root-most vertex of $b$.  This adds at most $\hat{\ep}^2D$ regret to any path, since each path has exactly one starting point.  The condensed branches are now all covered by \emph{detours}, so the reassignment argument is equivalent to {\sc Minimum Makespan Vehicle Routing}, adding at most an additional $\hat{\ep}^2D$ regret.  Without loss of generality, we assume that \emph{all} paths start at the \emph{top} (rootmost vertex) of some leaf cluster; paths are already assumed to start on the backbone, and the start point can always be extended to a descendant vertex without increasing the path's regret.

As in {\sc Minimum Makespan Vehicle Routing}, small clusters can be assigned (and charged) to descendant edge clusters such that the vehicle covering a given leaf cluster also covers at most two small clusters. We have now accounted for the single coverage of leaf and small clusters. 

Next, consider some edge cluster $C$.  $C$-passing and $C$-collecting routes here consist both of \emph{through paths} that start in a descendant of $C$ as well as \emph{through detours} that do not start in a descendant of $C$ and extend through the cluster to descendant clusters (see Figure~\ref{fig:regret_paths}). As before, regret can be reassigned so that each edge cluster is covered by at most one $C$-collecting route, adding at most $\frac{\hat{\ep}^2}{2}D$ regret to any path.

\begin{figure}[!h]
\captionsetup[subfigure]{justification=centering}
\centering
  \subfloat[\label{fig:regret_collecting}]
    {\includegraphics[width=0.3\textwidth]{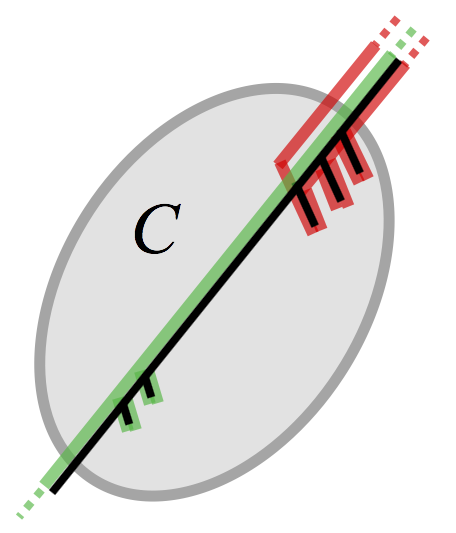}}
        \quad\vline\quad
  \subfloat[\label{fig:regret_through}]
    {\includegraphics[width=0.27\textwidth]{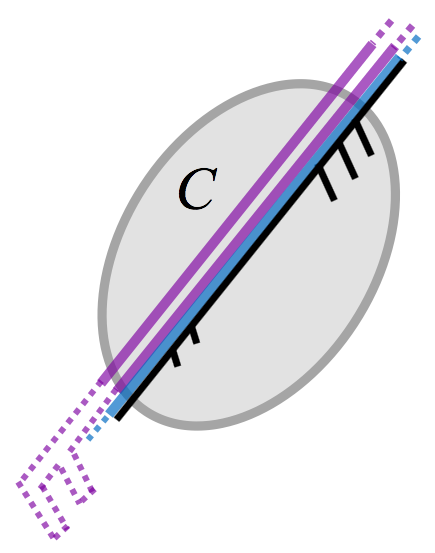}}
\caption{\label{fig:regret_paths} (a) The red path is a $C$-ending path and the green path is a $C$-collecting through path (b) The blue path is a $C$-passing through path and the purple path is a $C$-passing through detour.}
\end{figure}

Additionally, if the largest amount of regret accumulated in $C$ by a $C$-ending route is at most $\hat{\ep}D$, then regret can be reassigned to a single $C$-ending route.  Alternatively, if this amount is more than $\hat{\ep}D$ and is accumulated by $C$-ending route $t$, then all regret from other $C$-ending routes can be \emph{charged} to $t$'s regret in $C$, since this adds at most $\frac{\hat{\ep}^2}{2}D$ to the regret of $t$.

Using proofs analogous to those in Section~\ref{sec:makespan} we can show that there exists a near-optimal solution such that small cluster and edge clusters have single coverage, and edge clusters have split coverage.  All that remains to proving the structure theorem is to show that there exists such a solution such that each vehicle covers clients in a bounded number of clusters.

Bounding the number of leaf clusters and small clusters a vehicle covers is straightforward: leaf clusters have large regret so a vehicle can only afford to cover a small number of them, and each vehicle covers at most 2 small clusters for every leaf cluster.  Edge clusters with either single coverage, or \emph{balanced} split coverage (in which both vehicles cover a large fraction of the cluster's regret) are similarly straightforward to handle (see proof of Theorem~\ref{thm:structure_makespan}).  Moreover, if $C$ has split coverage and the $C$-ending path have very small regret (e.g. less than $\hat{\ep}^4D$), its clients can be reassigned and charged to the $C$-collecting path, which must have a relatively large regret and can afford to singly cover $C$.   

The difficult case is when $C$ has split coverage and the $C$-collecting path has very small regret in $C$. Unlike in Theorem~\ref{thm:structure_makespan}, even though the $C$-ending path has relatively large regret, it cannot afford to singly cover $C$ because the \emph{backbone} may be arbitrarily long.  We say that $C$ has \emph{dispersed} coverage in this case and try to bound the number of clusters with \emph{dispersed} coverage for which any single vehicle is a $C$-collecting path.

Consider some woolly edge $e$, and let $\{C_1,C_2,...,C_l\}$ be the clusters along $e$ that have dispersed coverage, ordered such that $C_{i+1}$ is rootward of $C_i$ for all $i$.  Though the $C_1$-ending path cannot afford to cover the clients covered by the $C_1$-collecting path in $C_1$, it \emph{can} afford to cover the clients covered by the $C_2$-collecting path in $C_2$, since it is necessarily passing through $C_2$ on the way to the root.  In general, the $C_i$-ending path can cover the clients covered by the $C_{i+1}$-collecting path in $C_{i+1}$. Finally, we can assign woolly edge $e$ to a descendant leaf cluster $C_L$ in such a way that each leaf cluster is assigned at most two woolly edges, and charge the regret of the $C_1$-collecting path in $C_1$ to the vehicle that covers $C_L$.  This \emph{propagation reassignment} only increases the maximum regret of a vehicle by a $(1+\hat{\ep})$ factor.  Now, for each cluster $C$ with dispersed coverage, the regret of the $C$-collecting path is charged to a subpath with relatively large regret in a descendant cluster, so we can bound the number of times this occurs for any given vehicle.

\subsection{Dynamic Program}

The dynamic program for the {\sc School Bus Routing} problem must account separately for the through paths and regret detours.  For through paths, a configuration stores the (rounded) regret already accumulated, whereas for regret detours a configuration stores the (rounded) total length.  Of note, a regret detour can be merged with a through path or another regret detour, but two through paths cannot be merged together.

The bigger challenge with adapting the DP is with route projection: a regret detour can be arbitrarily far from the root, so the DP cannot project the detour all the way to the root.  We address this by setting \emph{breakpoints} along the tree such that for every vertex $v$, there is an ancestor breakpoint within a distance $R$, and such that no two breakpoints are within $\ep R$ of each other.  Breakpoints can be defined greedily.  The DP then stores rounded regret detours projected to the next rootward breakpoint.  By construction, when two detours are merged, they agree on breakpoint, and twice the distance to the breakpoint can be subtracted, as with merging a detour and a through path.  When a breakpoint is encountered by the DP, all regret detours are projected to the next rootward breakpoint.  This adds $O(\ep R)$ rounding operations to any given path, and is therefore negligible.

\section{Generalizing to Multiple Depots}\label{sec:multi}

In this section, we sketch how our framework can be extended to allow for multiple depots.  We illustrate this extension using {\sc Minimum Makespan Vehicle Routing}, but note that it can be applied to other variants.  Here, each tour must start and end at the same depot, and a client can be covered by a tour from any depot.

\begin{theorem}
\label{thm:multi}
For every $\ep > 0$ and $\rho > 0$, there is a polynomial-time algorithm that, given an instance of {\sc $\rho$-Depot Minimum Makespan Vehicle Routing} on a tree, finds a solution whose makespan is at most $1+\ep$ times optimum.
\end{theorem}

\subsection{Structure Theorem}

Label the depots $\{r_1, r_2, ..., r_\rho\}$, and, without loss of generality, assume each depot is a leaf. and that the tree is rooted at some arbitrary vertex $r_0$, rather than at a depot.  The clustering is then done, with respect to $r_0$, identically to the single-depot setting, with $\delta= \hat{\ep}$ to be defined later, and with the load  $g(H)$ of subgraph $H$ again being defined to be twice the length of edges in $H$.

It is more convenient to assume that the depots are \emph{outside of} the clusters.  During the clustering process, if a depot $r_i$ is in a cluster $C$ with root-most vertex $v$, then a leaf edge $(v,r'_i)$ of length $d_T(r_i,v)$ is added to $v$, padded with zero length edges, to maintain a binary tree.  The depot is then assumed to be located at $r'_i$, and thus outside of $C$ (see Figure~\ref{fig:multi_clustering}).  This modification adds at most $\hat{\ep} D$ to the optimal makespan, to account for the extra distance that tours from $r_i$ must travel to reach $r'_i$. Otherwise, the increase to makespan caused by condensing is equivalent to Lemma~\ref{lem:condense} of the single-depot setting.

To show that requiring all small clusters to have single coverage does not increase the optimal makespan by too much is more challenging than in the single-depot case.  Recall that with only one depot, the load from a small cluster $C_S$ could be charged to a descendant leaf cluster $C_L$, because the tour $t_L$ that covers $C_L$ \emph{must} pass through $C_S$ on the way to the root depot.  This is not true for the multi-depot setting.

Observe though that if $C_S$ is incident to a \emph{woolly} branch $b$ (a branch of the condensed tree) that contains no depots, then the load from $C_S$ \emph{can} be charged to a leaf cluster $C_L$ in $b$, since the tour covering $C_L$ must leave branch $b$ to reach its depot, and therefore either pass through or abut $C_S$.  Such small clusters can be assigned to such leaf clusters so that each leaf cluster is charged for at most three small clusters, adding at most $4\hat{\ep}D$ to the makespan, as in Lemma~\ref{lem:small}.  

What remains are the small clusters \emph{not} incident to such depot-free branches.  We claim that there are at most $2\rho$ such clusters.  Note that since we moved depots to be outside of clusters, the cluster tree $T^*$, resulting from contracting clusters, now has vertices for each cluster, depot, branching point, and the root.  Consider the \emph{meta} tree structure $T^*_{meta}$ induced by the depots and the root (see Figure~\ref{fig:multi_meta}).  There are fewer than  $2\rho$ edges of $T^*_{meta}$ since it is a binary tree with depot leaves.  Consider some edge $e_{meta}$ of $T^*_{meta}$.  If $e_{meta}$ \emph{contains} more than one small cluster (i.e. corresponds to a subgraph of $T^*$ with more than one small cluster vertex), those small clusters must be incident to depot-free branches. Otherwise, $e_{meta}$ contains at most one small cluster.  Therefore, there are at most $2\rho$ such small clusters, and requiring them to have single coverages increases optimal makespan by at most $\rho\hat{\ep}^2D$.

\begin{figure}[!h]
\centering
\includegraphics[width=0.6\textwidth]{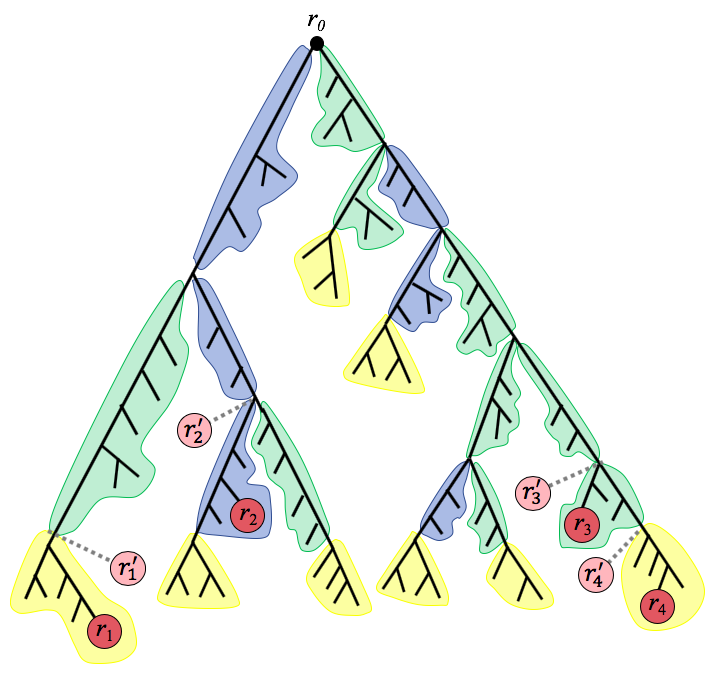}
\caption{\label{fig:multi_clustering} A four-depot clustering example showing depots moved outside of clusters.  Leaf clusters are yellow, small clusters are blue, and edge clusters are green.}
\end{figure}

\begin{figure}[!h]
\captionsetup[subfigure]{justification=centering}
\centering
  \subfloat[\label{fig:multi_t_star}]
    {\includegraphics[width=0.35\textwidth]{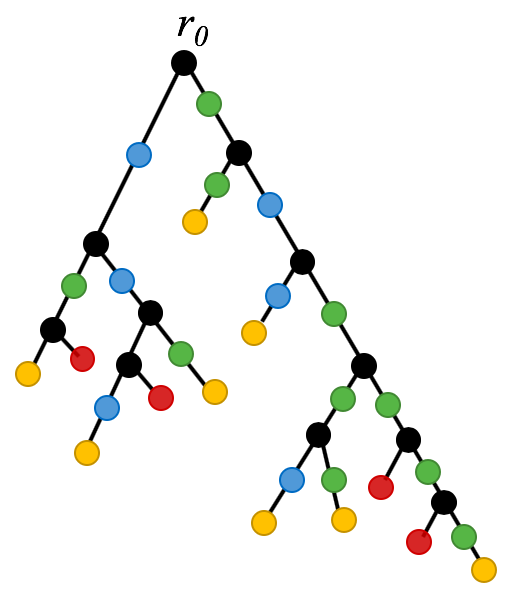}}
        \quad\vline\quad
  \subfloat[\label{fig:multi_meta}]
    {\includegraphics[width=0.4\textwidth]{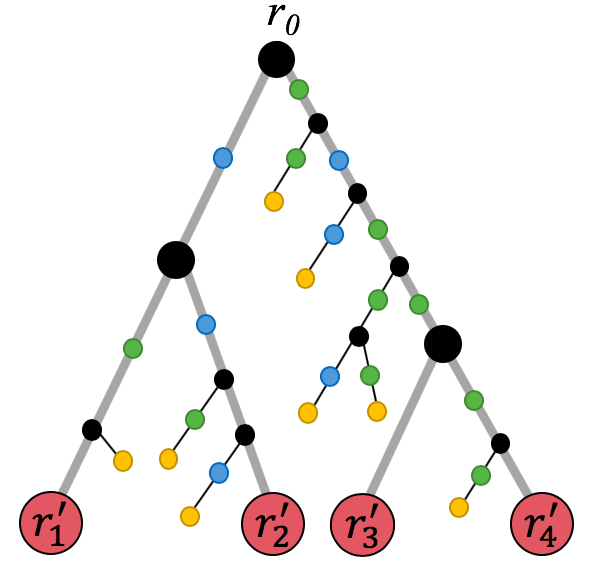}}
\caption{\label{fig:multi_depots} (a) Cluster tree $T^*$ corresponding to Figure~\ref{fig:clustering}; (b) Corresponding meta tree $T^*_{meta}$.  Meta edges are shown as thick gray lines;}
\end{figure}

Last, the definition of split coverage needs to be adjusted to accommodate multiple depots.  In particular an edge cluster $C$ may have $C$-ending tours from \emph{each} end.  We let \emph{top} $C$-ending tours denote those that enter $C$ from the root-ward end and \emph{bottom} $C$-ending tours denote those that enter $C$ from the leaf-ward end.  In the multi-depot setting, the analogous definition of split coverage is that the leaves of $C$ along the backbone can be partitioned into three continuous intervals, such that if the root-most interval is nonempty it is covered by a single top $C$-ending tour, if the middle interval is nonempty it is covered by a single $C$-collecting tour, and if the leaf-most interval is nonempty it is covered by a single bottom $C$-ending tour.  Using analogous analysis to Lemma~\ref{lem:split} shows that requiring edge clusters to have split coverage adds at most $5\hat{\ep} D$ to the optimal makespan.

It can then be shown that there is a solution of makespan  $(1+O(\hat{\ep} + \rho\hat{\ep}^2))D$ such that each tour covers clients in $O(\frac{1}{\hat{\ep}^3} + \rho)$ clusters.  To prove the structure theorem, it then suffices to set $\hat{\ep} = \frac{\ep}{c\rho}$ for an appropriate constant $c$.

\subsection{Dynamic Program}

To adapt the dynamic program for multiple depots, we first note that the configurations must now account for tours \emph{to each depot}. The challenge is that our single-depot dynamic program depends on  projecting vehicle routes \emph{up} the tree, toward a common depot.  In particular, this was the same direction as the DP traversal itself.  In the multiple-depot setting, the DP must move \emph{away from} depots, but it cannot afford to guess which clients to project a tour toward.  To address this issue, for each depot $r_i$ we assign a \emph{sub-root} $\bar{r}_i$ to be the root-most vertex that is within a distance $D/2$ of $r_i$.  In particular, this sub-root is the root-most vertex that can possibly be included in any tour from $r_i$.  

When the DP introduces a new tour to depot $r_i$, it is then projected \emph{up} to $\bar{r}_i$ instead of down to $r_i$.  This includes when the DP processes $r_i$ itself: a configuration describes how many tours start (and end) at $r_i$, and they are all projected from $r_i$ up to $\bar{r}_i$ (see Figure~\ref{fig:multi_projection}).  Two depot-$r_i$ tour segments $t_1$ and $t_2$ can only be merged at some vertex $v$ if $l(t_1)+l(t_2) - 2d_T(v,\bar{r}_i) \leq (1+\ep)D$.  After merging, they are stored as a depot-$r_i$ tour segment of length $l(t_1)+l(t_2) - d_T(v,\bar{r}_i)$ rounded up to the nearest $\hat{\ep}^{4}D$ (see Figure~\ref{fig:multi_merge}).  If at some vertex $v$, the (rounded) projected tour length reaches $d_T(v,\bar{r}_i) + (1+\ep)D$, the projected subtour from $v$ to $\bar{r}_i$ is removed and that the tour ends at $v$ (see Figure~\ref{fig:multi_cap}). In other words, $v$ is the most root-most vertex reached by that tour. The overall runtime is $n^{O(\rho\ep^{-8})}$.

\begin{figure}[!h]
\captionsetup[subfigure]{justification=centering}
\centering
  \subfloat[\label{fig:multi_projection}]
    {\includegraphics[width=0.25\textwidth]{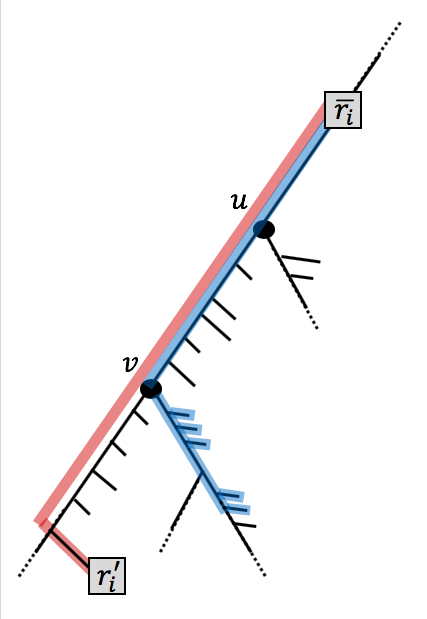}}
        \quad\vline\quad
  \subfloat[\label{fig:multi_merge}]
    {\includegraphics[width=0.25\textwidth]{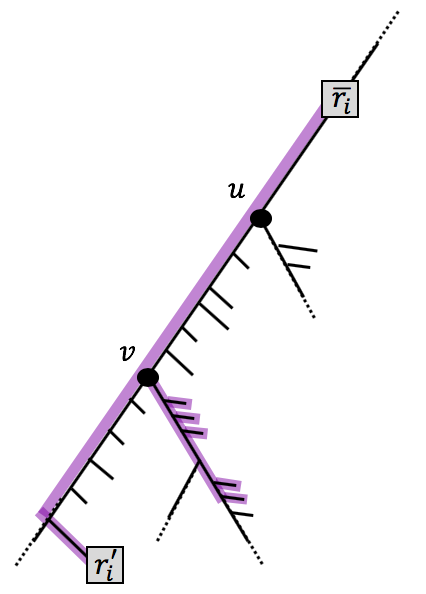}}
        \quad\vline\quad
  \subfloat[\label{fig:multi_cap}]
    {\includegraphics[width=0.25\textwidth]{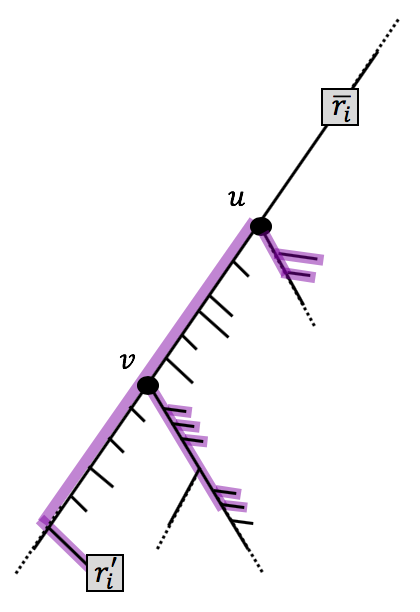}}
\caption{\label{fig:multi_depots} (a) Tours to $r_i$ are projected to $\bar{r}_i$; (b) Tours to $r_i$ from (a) merged at vertex $v$ and projected to $\bar{r}_i$; (c) Tour to $r_i$ capped at vertex $u$ once tour length reaches $(1+\ep)D$;}
\end{figure}

\section{Conclusion}\label{sec:conclusion}

In this paper, we present a general framework that yields PTASs for several classic vehicle routing problems on trees including {\sc Minimum Makespan Vehicle Routing}, {\sc Distance-Constrained Vehicle Routing}, {\sc Capacitated Vehicle Routing}, and {\sc School Bus Routing}. The breadth of these problems highlights the  flexibility and usefulness of the framework, especially given its ability to capture distance-based, capacity-based, and regret-based load functions. Further, we believe this framework can serve as a tool for applications beyond the scope of this paper. 

%
%
\bibliographystyle{plain}
\bibliography{main}

\appendix

\section{Chen and Marx's Previous Claim to a Minimum Makespan PTAS.}
\label{app:chen_marx}

The intuition behind our result is similar to that of~ \cite{chen_marx}: since the main challenge arises from having to account for small tour segments, modify the instance so that only large tour segments need to be accounted for.  Because each tour can only have a small number of these larger tour segments, the algorithm can essentially \emph{guess} all possible ways that the segments can be joined to form tours.  The difficulty with this approach is in showing that this type of grouping of small segments into large segments can be done without incurring too much error.
As did~\cite{chen_marx}, we also observed that assuming each subtree with length less than $\ep OPT$ to be covered by a single tour increases the makespan by $O(\ep OPT)$, where $OPT$ is the optimum makespan.

The presented argument of Theorem 2.1~\cite{chen_marx} (which their PTAS requires) depends on the following argument, though they use different terminology: Let $T_v$ denote the subtree rooted at $v$.  The input can be \emph{safely} modified\footnote{\emph{Safe} graph modifications are those that result in an equivalent instance and consist  of subdividing a single edge $e_0$ into two smaller edges $e_1$ and $e_2$ such that $\ell(e_1) + \ell(e_2) = \ell(e_0)$.} so that there exists a set of vertices $CR$ such that the subtrees rooted at vertices in $CR$ collectively \emph{partition} the leaves of the tree and such that $CR$ has the following properties:

\begin{enumerate} 
\item For each $v\in CR$, $T_v$ is \emph{small}: for each child $u$ of $v$, the length of $T_u$ is less than $\ep OPT$.
\item For each $v\in CR$, $T_v$ is \emph{large}: the length of $T_v$ is at least $\ep OPT$.
\item The vertices of $CR$ are \emph{independent}: no two distinct vertices in $CR$ have an ancestor-descendant relationship.
\end{enumerate}

Each of these properties is vital to the proof of Theorem 2.1~\cite{chen_marx}: Property 1 bounds the solution error, Property 2 ensures polynomial runtime, and Property 3 ensures solution feasibility.

Given such a set $CR$, Chen and Marx~\cite{chen_marx} prove that there exists a near-optimal solution in which each \emph{branch} of a subtree rooted at a vertex in $CR$ is covered by a single tour.  They do so by showing that an optimal solution $SOL$ can be transformed into a near-optimal solution $SOL'$ in which each such branch is wholly covered by one of the tours that partially covers it in $SOL$.  This proof depends on Property 1 to bound the resulting error of this transformation and on Property 3 to maintain solution feasibility while independently reassigning branch coverage.  A dynamic programming strategy can then be used to find such a near-optimal solution, which requires Property 2 in order to achieve a polynomial runtime.

The problem, though, is that no such set $CR$ is guaranteed to exist. In the process of modifying the input graph to ensure that their proposed set $CR$ satisfies Property 3, they end up with a set that no longer satisfies Property 2.  

In fact, it is not difficult to show the following.

\begin{lemma}
\label{lem:chen_marx}
There exist instances in which no such set $CR$ as described above exists.
\end{lemma}
\begin{proof}
We describe an instance that neither admits such a set $CR$ nor can be modified using \emph{safe} graph modifications into an equivalent instance that admits such a set. Let $T$ be a tree with root $r$ that consists of a \emph{central} path $P = r, v_1, v_2, ..., v_{l}$, \emph{side} subtrees $T_i$ for $i \in \{1,2,...,l\}$ such that the root of $T_i$ is a child of $v_i$, and \emph{main} subtree $T_{l+1}$ whose root $v_{l+1}$ is a child of $v_l$. Let edges $(v_i,v_{i+1})$ along $P$ have small nonzero lengths, let each side subtree $T_i$ have total length less than $\frac{\ep OPT}{2}$ and let the main subtree $T_{l+1}$ have total length at least $\frac{\ep OPT}{2}$ (see Figure~\ref{fig:counter}).  In order to satisfy Property 2, no vertex $v$ in any side subtree $T_i$ for $0<i\leq l$ can be in $CR$, since the subtree rooted at such a $v$ is not large enough.  So, since the subtrees rooted at vertices in $CR$ partition the leaves of $T$, at \emph{least} one vertex in $P$ must be in $CR$ to cover the leaves of trees $T_i$.  At the same time, in order to satisfy Property 3, at \emph{most} one vertex in $P$ can be in $CR$.  But, since main subtree $T_{l+1}$ is large, no choice of vertex in $P$ can satisfy Property 1.

Note that the side and main subtrees are general, and can be assumed to be defined after any safe modifications occur.  Additionally, if any edges of $P$ are subdivided, the argument above clearly still holds.

\begin{figure}[!h]
\centering
      \includegraphics[width=.5\textwidth]{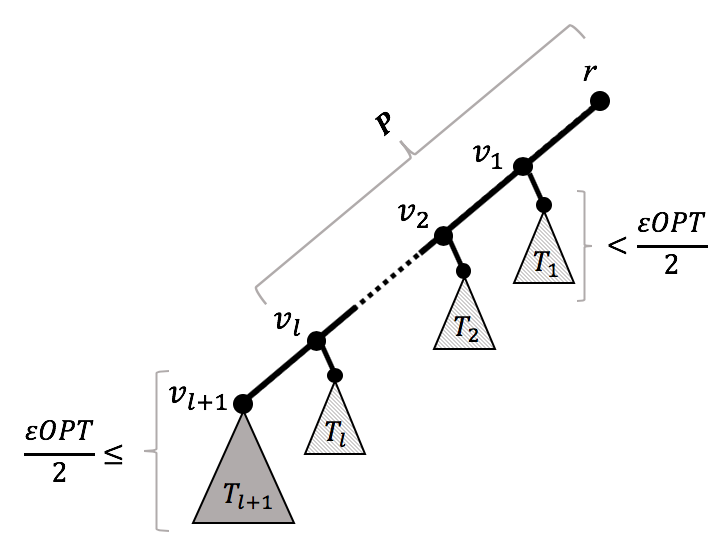}
\caption{\label{fig:counter} Counterexample to~\cite{chen_marx}.}
\end{figure}
\end{proof}

The counterexample in the above proof allows an arbitrarily complex and heavy main subtree $T_{l+1}$, so it cannot be argued that only trivial examples fail to admit a set $CR$.  In particular, this pattern of light side subtrees can continue throughout the main subtree, so these \emph{difficult areas} cannot be easily dealt with separately.

To the best of our knowledge, there is no straightforward way to salvage the theorem used in \cite{chen_marx}.  In fact, the crux of the problem is highlighted by the \emph{absence} of such a set $CR$: How can the input be simplified so as to balance error with runtime given that coverage choices are \emph{not} independent?  Our approach uses similar intuition to that of~\cite{chen_marx}, but manages to successfully address the critical challenges of the problem where prior attempts have fallen short.

Rather than partition the tree uniformly into subtrees, we recognize that the behavior of solutions \emph{looks} differently near the leaves than it does along internal vertices.  As such, our framework partitions the \emph{entire} tree into three different types of \emph{clusters} (see Figure~\ref{fig:clusters}).  \emph{Leaf clusters} are defined greedily so as to satisfy properties similar to 1-3 above and serve to anchor the cluster structure.  \emph{Small clusters} are internal clusters with weight (and frequency) small enough to be effectively ignored while incurring only a small error that can be \emph{charged} to the leaf clusters.  The remainder of the tree is grouped into \emph{edge clusters}, with properties similar to 1 and 2 above: their weight is \emph{small} enough to assume, without incurring too much error, a \emph{simple} structure to the way that tours cover them, yet they are \emph{large} enough so that any tour can only cover a bounded number of them. Though the edge clusters themselves do \emph{not} maintain independence, we show a weaker, sufficient property exhibited by tour segments in these clusters that ensures tour connectivity while treating edge clusters independently.

\end{document}